\documentclass[reqno, 11pt]{amsart}

\usepackage[letterpaper,hmargin=1in,vmargin=1.2in]{geometry}
\usepackage{amsmath, amssymb, amsthm, verbatim, url}
\usepackage{graphicx}
\usepackage{amsmath,amssymb,amsthm,mathrsfs,graphicx,url}
\usepackage[usenames,dvipsnames]{color}

\DeclareMathOperator \re {Re}
\DeclareMathOperator \im {Im}

\newtheorem{thm}{Theorem}
\newtheorem{lem}{Lemma}

\title{Semiclassical resonance asymptotics for the  delta potential on the half line}

\author{Kiril Datchev}
\address{Department of Mathematics, Purdue University, West Lafayette,
  IN 47907-2067}
\email{kdatchev@purdue.edu}
\author{Nkhalo Malawo}
\address{Department of Mathematics, Purdue University, West Lafayette,
  IN 47907-2067}
\email{nkhalom@protonmail.com}

\thanks{K. Datchev was supported in part by NSF grant DMS-1708511. N. Malawo was supported in part by an REU Grant from the Purdue Math Department Tong Endowment. The authors are grateful to Jeffrey Galkowski and Maciej Zworski for helpful discussions, and also to the anonymous referees for their comments and corrections.}

\begin{document}

\begin{abstract}
We compute resonance width asymptotics for the delta potential on the half-line, by deriving a formula for  resonances in terms of the Lambert W function and applying a series expansion. This potential is a simple model of a thin barrier, motivated by physical problems such as quantum corrals and leaky quantum graphs.
\end{abstract}

\maketitle

\section{Introduction}
The analysis of scattering by \textit{thin barriers} is important for many physical problems, including quantum corrals \cite{bzh} and leaky quantum graphs \cite{ex}. In \cite{gs15, g15, g16,g19book,g19} Galkowski--Smith and Galkowski study the distribution of resonances for operators of the form
\[
-h^2 \Delta + h^{2} h^{-\alpha}  \delta_{\partial \Omega},
\]
where $h>0$ is a semiclassical parameter, $\Delta$ is the Laplacian on $\mathbb R^n$, and $\delta_{\partial \Omega}$ is a delta function on the boundary of an open bounded set $\Omega$ with smooth boundary $\partial \Omega$. The factor of $h^{-\alpha}$ models a barrier whose interaction with waves depends on frequency (having a nontrivial dependence is typical in physical systems: see \cite{bzh}), with the positive parameter $\alpha$ determining how quickly the strength of the barrier grows with frequency.

In this paper we consider the corresponding operator on the half line $(0,\infty)$, namely
\begin{equation}\label{e:delta1}
-h^2 \partial_x^2 + h^{2} h^{-\alpha} \delta_1,
\end{equation}
where $\delta_1$ is the Dirac delta function centered at $x=1$,
with Dirichlet boundary condition at $x=0$. In this setting we can analyze scattering more simply, fully, and precisely than is possible in the more complicated higher-dimensional case.

A number $z \in \mathbb C \setminus \{0\}$ is a \textit{resonance} of \eqref{e:delta1} if and only if there is a continuous function $u$ such that
\begin{equation}\label{e:ressyst} \begin{split}
&u(x) = \sin(zx/h) \quad \text{ for } \quad 0<x<1, \\
&u(x) = C e^{izx/h} \quad\quad  \text{ for } \quad x> 1, \\
&u'(1-) - u'(1+) + h^{-\alpha} u(1) = 0,
\end{split}\end{equation}
for some constant $C$, where $u'(1-)$ and $u'(1+)$ are respectively the derivatives from the left and from the right of $u$ at $x=1$.
A solution to \eqref{e:ressyst} exists if and only if
\[
\sin(z/h) = C e^{iz/h}, \qquad (z/h)\cos(z/h) - (iz/h)Ce^{iz/h} + h^{-\alpha} C e^{iz/h} = 0,
\]
or, equivalently,
\begin{equation}\label{e:reseq}
 h^{-\alpha} e^{2iz/h} - h^{-\alpha}  + 2iz/h = 0.
\end{equation}

The values of $\im z$ for solutions to \eqref{e:reseq} are called the resonance widths, and they give the rates of decay of waves. Our main results establish semiclassical asymptotics for all resonance widths, apart from resonances near zero and near infinity. The behavior is logarithmic when $\alpha<1$ and polynomial when $\alpha>1$.

\begin{figure}[h]
    \centering
    \includegraphics[width=12cm]{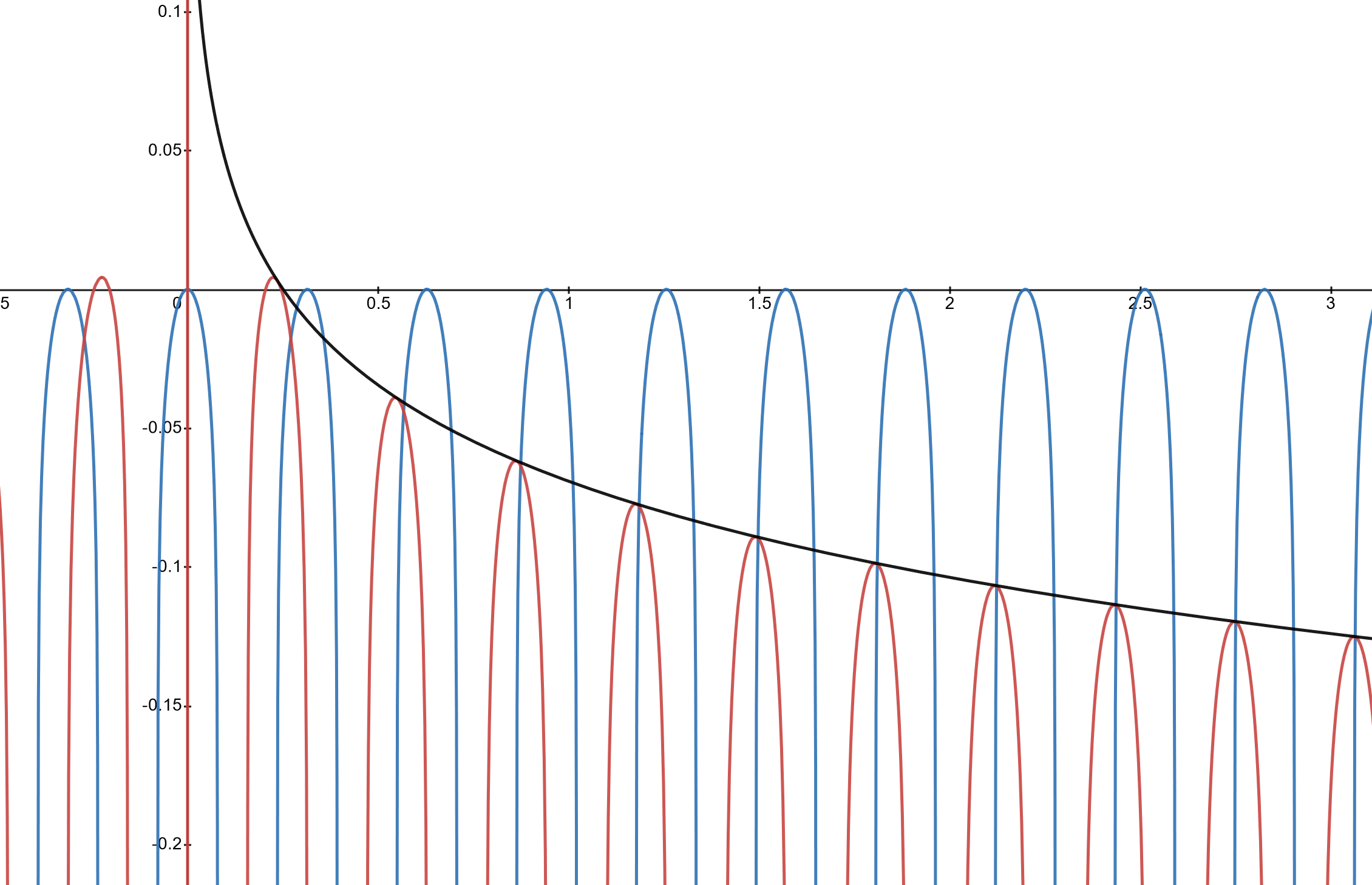}
    \caption{This figure illustrates Theorem \ref{t:asmall} with $h = 0.1$ and $\alpha = 0.7$. The black curve is the leading approximation $-\im z = \frac h 2 \ln|2h^{\alpha-1}\re z|$ and the blue and red curves are the real and imaginary parts of \eqref{e:reseq}. The resonances occur where the blue and red curves intersect. }
    \label{fig:my_label}
\end{figure}

\begin{thm}\label{t:asmall}
Let $\alpha \in (0,1)$ and $ \varepsilon \in (0,1)$ be given. Then there is $h_0>0$ such that, when $h \in (0,h_0]$, all solutions to \eqref{e:reseq} satisfying
\begin{equation}\label{e:zeps}
\varepsilon \le |z| \le 1/\varepsilon,
\end{equation}
obey
\begin{equation}\label{e:tasmall}
0 \le -\im z - \frac h 2 \ln|2h^{\alpha-1}\re z| \le \frac 5 4   h^{3-2\alpha}\varepsilon^{-2}
\end{equation}
\end{thm}

\begin{thm}\label{t:abig}
Let $\alpha >1$ and  $ \varepsilon \in (0,1)$ be given. Then there is $h_0>0$ such that, when $h \in (0,h_0]$, all solutions to \eqref{e:reseq} satisfying \eqref{e:zeps}
obey
\begin{equation}\label{e:tabig}
|\im z +   (\re z)^2 h^{2\alpha-1}| \le 7 h^{2\alpha+1} \ln^2(h^{-\alpha}) +  34 \varepsilon^{-4}h^{4\alpha - 3}.
\end{equation}
\end{thm}

\begin{figure}[h]
    \centering
    \includegraphics[width=12cm]{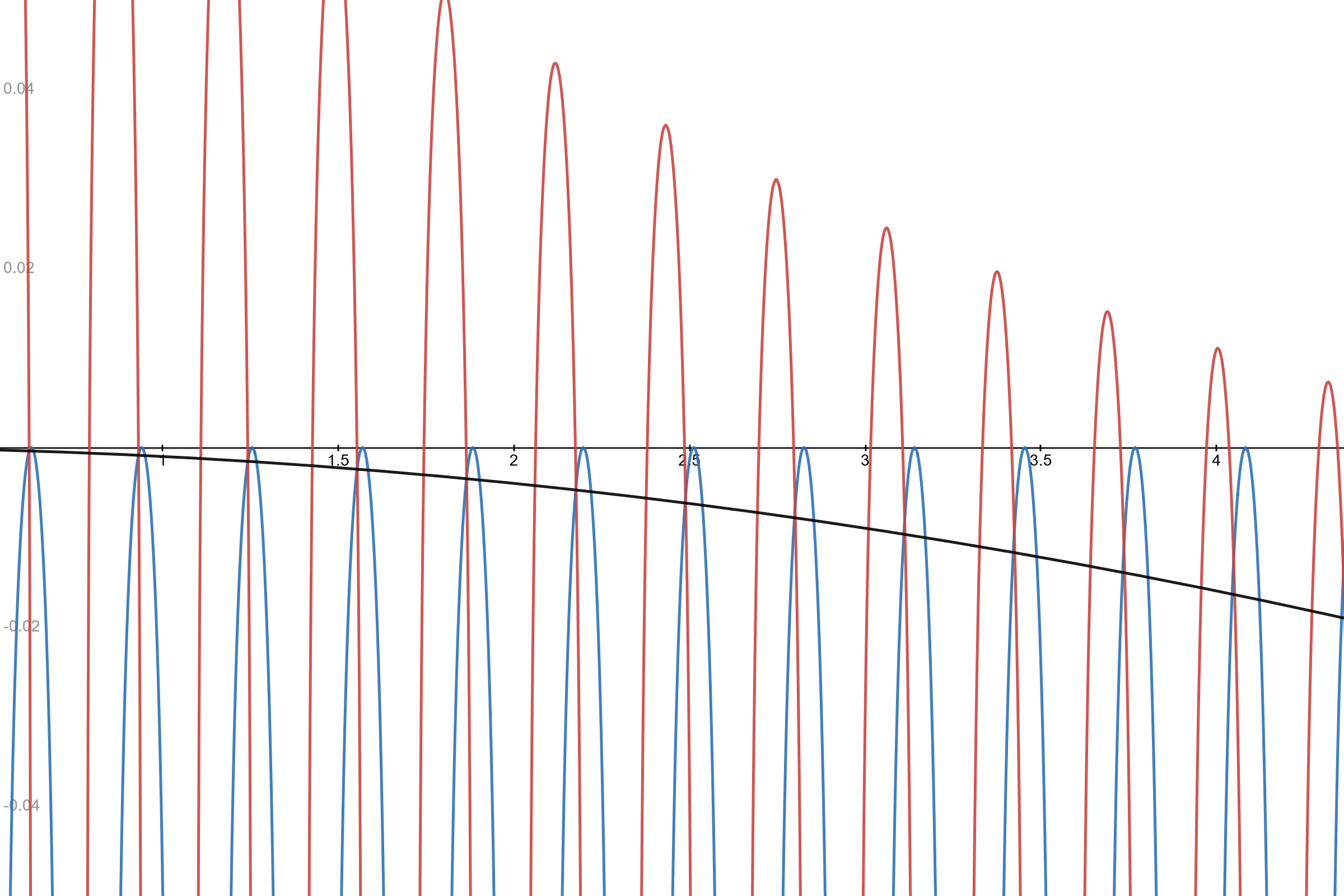}
    \caption{This figure illustrates Theorem \ref{t:abig} with $h = 0.1$ and $\alpha = 2$. The black curve is the leading approximation $-\im z = (\re z)^2 h^{2\alpha-1}$ and the blue and red curves are the real and imaginary parts of \eqref{e:reseq}. The resonances occur where the blue and red curves intersect. }
    \label{fig:my_label2}
\end{figure}

We have not attempted to optimize the numerical constants; the main point is the dependence on $h$ and $\alpha$. 
The asymptotics \eqref{e:tasmall} and \eqref{e:tabig} refine lower bounds obtained for much more general problems by Galkowski: see in particular  \cite[Section 1.B]{g19}. In our simpler situation, we obtain a formula for the resonances in terms of the Lambert W function in equation \eqref{e:zk} below. Applying the known series for this function (see equations \eqref{e:wk} and \eqref{e:rk} below) leads to a full convergent series expansion for the resonances; thus our methods could be elaborated to give any number of further terms in the expansions  \eqref{e:tasmall} and \eqref{e:tabig}, giving the resonance widths to any desired accuracy.

To compare our results with the higher-dimensional setting  studied in \cite{g16, g19}, let us focus on the case $ \re z \sim 1$. Then our results give widths of order $\frac h 2 \ln(2h^{\alpha-1})$ when $\alpha<1$ and $h^{2\alpha -1}$ when $\alpha > 1$: stronger point interaction, i.e. a larger factor in front of $\delta_1$ in \eqref{e:delta1}, leads to slower decay. By contrast, Galkowski's results for a thin barrier supported on a circle in two-dimensions give widths of order $\frac h 2 \ln(2h^{\alpha-1})$ when $\alpha < 5/6$ and $h^{2\alpha - 2/3}$ when $5/6 \le \alpha \le 1$; see Theorems 1 and 2 of \cite{g16}. The key difference is that the shift between logarithmic and polynomial behavior occurs at $\alpha=1$ in the one-dimensional setting, and at $\alpha = 5/6$ in the higher-dimensional setting. As explained in \cite{g16, g19}, resonances polynomially close to the real axis when $\alpha >5/6$ arise from singularities propagating along glancing rays. 

In our one-dimensional setting, there are no such glancing rays, and resonances polynomially close to the real axis are caused by reflection. The reason for the change of behavior between $\alpha<1$ and $\alpha>1$ is most easily seen by computing the reflection coefficient of the model operator $-h^2\partial_x^2 + h^2 h^{-\alpha} \delta_0$, on the real line, at energy $z^2$, which is 
\[
R =  \frac{h^{2-2\alpha}}{4z^2 + h^{2-2\alpha}};
\]
see e.g. Section 2.5 of \cite{grif}. As $h \to 0$, we have $R \to 0$ if $\alpha <1$, and $R \to 1$ if $\alpha >1$. Thus there is a small amount of reflection when $\alpha<1$ and a large amount when $\alpha>1$.

The operator we are studying can be related to an operator of the form $-\partial_y^2 - \delta_a$ by rescaling. More specifically, put $y = h^{-\alpha}x$. Then
\[
(-h^2 \delta_x^2 + h^{2-2\alpha} \delta_1 - z^2) = h^{2-2\alpha}(-\partial_y^2 - \delta_a - z_y^2), \qquad 
\text{where } \quad a = h^{-\alpha}, \quad z_y = h^{\alpha-1}z.
\]
This scaling also highlights the role of the transitional value $\alpha=1$.

The asymptotics in the case $\alpha =1$ are more complicated. Figure \ref{fig:a1} suggests that the approximation of Theorem 1 is more accurate at high energies and the approximation of Theorem 2 is more accurate at low energies. In both cases the resonance widths have size $h$ but the $z$ dependence is not the same. To get a precise asymptotic for $\alpha=1$ would require a finer analysis than is done below. Specifically, it would be necessary to distinguish high and low energies, as the leading behavior appears to be algebraic in $\re z$  when $\re z$ is small and logarithmic when $\re z$ is large.

\begin{figure}[h]
    \centering
    \includegraphics[width=12cm]{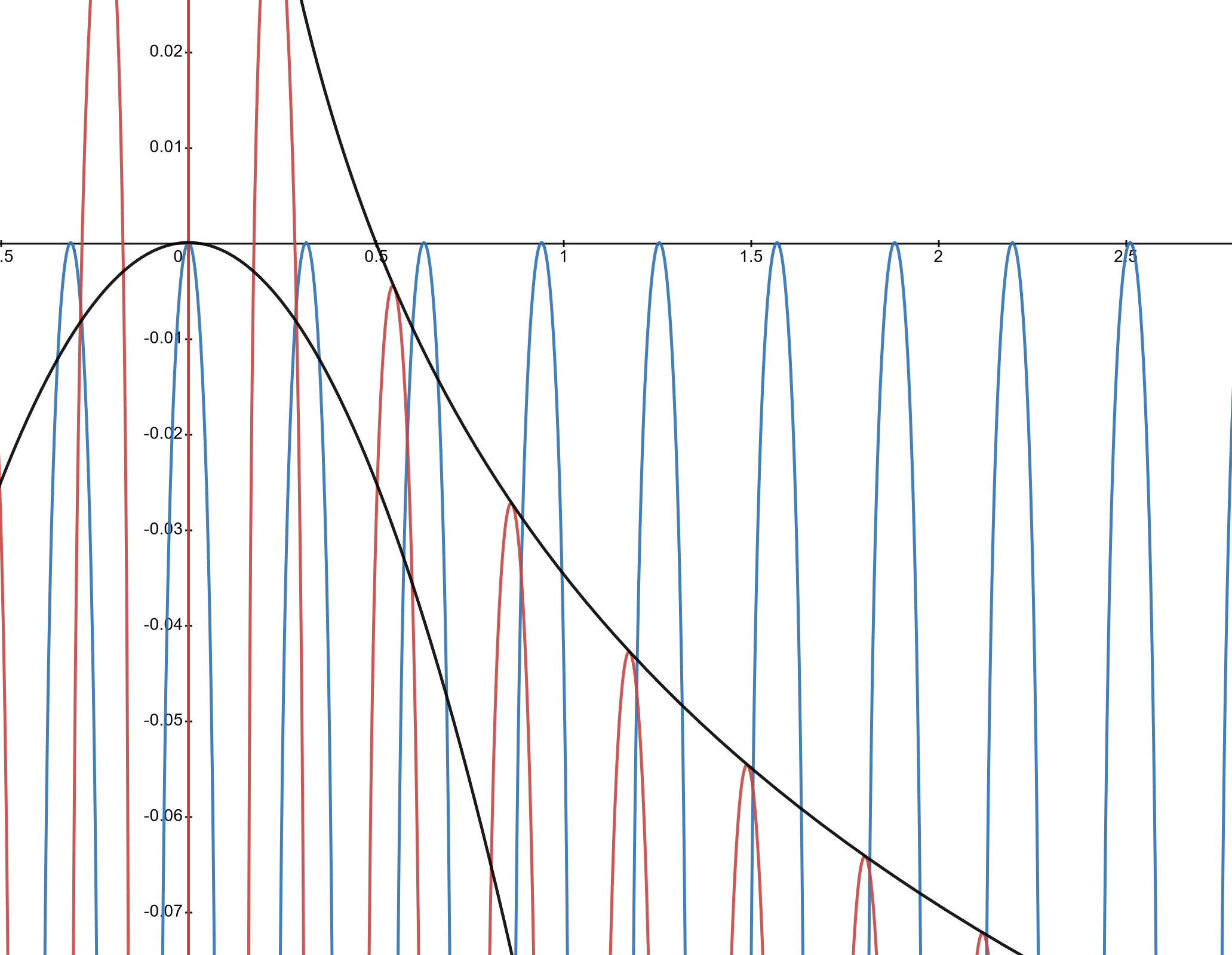}
    \caption{This figure illustrates the case $\alpha=1$ with $h = 0.1$. The lower black curve is the  approximation $-\im z = (\re z)^2 h^{2\alpha-1}$ and the upper black curve is the  approximation $-\im z = \frac h 2 \ln|2h^{\alpha-1}\re z|$. The blue and red curves are the real and imaginary parts of \eqref{e:reseq}. The resonances occur where the blue and red curves intersect. Note that the lower black curve more closely approximates the resonances at lower energies. As the energy increases, the upper black curve is a better approximation. }
    \label{fig:a1}
\end{figure}

In \cite[Theorem 2.6]{g15}, Galkowski studies a variant of our problem with $h^{-\alpha}\delta$ replaced by $h^{-\alpha}\delta'$, the derivative of the Dirac delta function. In that case the decay rates are of order $\frac{1+\alpha}2 h \ln(1/h)$ when $\alpha < -1$ and of order $h^{3+2\alpha}$ when $\alpha >-1$. Other examples of strings of resonances on such logarithmic and polynomial curves can be found in \cite{r,z,b,dkk,bfrz,g17,hw}. The large argument expansion of the Lambert W function has been previously used to compute resonance asymptotics in \cite{dkk,d,c}, and it has been used for Schr\"odinger equations with Dirac delta functions in \cite{ks,hm,s}. The closest of these papers to ours are Herbst and Mavi's \cite{hm} and Sacchetti's \cite{s}; the main difference is that the results in \cite{hm,s} correspond to the asymptotic regime $z \to 0$ where the difference between $\alpha<1$ and $\alpha>1$ is not visible.  See the discussion following equation \eqref{e:zk} for more. See \cite[Chapter II.2]{aghkh} and references therein for more general results concerning the resonances of finitely many $\delta$- and $\delta'$-interactions on the real line, and see \cite{dz} for a broader introduction to resonances.

\section{The Lambert W function}

We begin by reviewing some needed facts about the Lambert W function from \cite{cghjk}. This is the function which solves the equation
\begin{equation}\label{e:xexy}
x e^x = y.
\end{equation}
For any $y$, \eqref{e:xexy} has a countable infinity of complex solutions, $x = W_k(y)$, where  $k$ varies over $\mathbb Z$. We will be studying solutions to \eqref{e:xexy} with $y \gg 1$. Then $W_k(y)$ can be expanded in a convergent series, with coefficients given  in terms of the Stirling numbers $[{p \atop q}]$, where $[{p \atop q}]$ is the number of ways to arrange $p$ objects into $q$ cycles (see \cite[Chapter 6]{gkp}). More specifically, by equations (4.18), (4.19), and (4.20) of \cite{cghjk}, we have
\begin{equation}\label{e:wk}
W_k(y)=\ln(y)+2\pi ik-\ln(\ln(y)+2\pi ik)+R_k,
\end{equation}
where
\begin{equation}\label{e:rk}
R_k= \sum_{j=0}^\infty \sum_{m=1}^\infty c_{j,m} \frac{\ln^{m}(\ln(y)+2 \pi ik)}{(\ln(y)+2 \pi ik)^{j+m}} = \frac{\ln(\ln(y)+2\pi ik)}{\ln(y)+2\pi ik} + \cdots,
\end{equation}
and
\[
c_{j,m} = \frac 1 {m!} (-1)^j \left[{j+m \atop j+1} \right].
\]
Here and below the branch of $\ln$ is taken so that $-\pi  < \im \ln z \le \pi$ for all $z \in \mathbb C \setminus \{0\}$.

\begin{lem}
The series \eqref{e:rk} is absolutely convergent for $y$ large enough and $k \in \mathbb Z$. More precisely we have the tail estimate
\begin{equation}\label{e:tailest}
\begin{split}
\left| R_k - \frac{\ln(\ln(y)+2\pi ik)}{\ln(y)+2\pi ik}\right| &\le \sum_{j \ge 0, \ m \ge 1, \ (j,m) \ne (0,1)} \left| c_{j,m} \frac{\ln^{m}(\ln(y)+2 \pi ik)}{(\ln(y)+2 \pi ik)^{j+m}}\right|\\
&\le 2 \left|\frac{\ln(\ln(y)+2 \pi ik)}{\ln(y)+2 \pi ik}\right|^2,
\end{split}
\end{equation}
and
\begin{equation}\label{e:lnlnln}
\left|\frac{\ln(\ln(y)+2 \pi ik)}{\ln(y)+2 \pi ik}\right| \le \frac 12,
\end{equation}
for $y$ large enough.
\end{lem}

\begin{proof}
We have
\begin{equation}\label{e:twosums}
\begin{split}
\left| R_k - \frac{\ln(\ln(y)+2\pi ik)}{\ln(y)+2\pi ik}\right| &\le \sum_{j \ge 0, \ m \ge 1, \ (j,m) \ne (0,1)} \left| c_{j,m} \frac{\ln^{m}(\ln(y)+2 \pi ik)}{(\ln(y)+2 \pi ik)^{j+m}}\right|  \\
&= \sum_{m=2}^\infty \left| c_{0,m} \frac{\ln^{m}(\ln(y)+2 \pi ik)}{(\ln(y)+2 \pi ik)^{m}}\right| + \sum_{m=1}^\infty\sum_{j=1}^\infty \left| c_{j,m} \frac{\ln^{m}(\ln(y)+2 \pi ik)}{(\ln(y)+2 \pi ik)^{j+m}}\right|. 
\end{split}
\end{equation}
Now use $c_{0,m} = 1/m$ and \eqref{e:lnlnln} to write
\begin{equation}\label{e:j=0}
\sum_{m=2}^\infty \left| c_{0,m} \frac{\ln^{m}(\ln(y)+2 \pi ik)}{(\ln(y)+2 \pi ik)^{m}}\right| \le \frac 12 \sum_{m=2}^\infty \left|\frac{\ln(\ln(y)+2 \pi ik)}{\ln(y)+2 \pi ik}\right|^m \le \left|\frac{\ln(\ln(y)+2 \pi ik)}{\ln(y)+2 \pi ik}\right|^2.
\end{equation}
For the other term, write
\begin{equation}\label{e:jge1}
\begin{split}
\sum_{m=1}^\infty\sum_{j=1}^\infty& \left| c_{j,m} \frac{\ln^{m}(\ln(y)+2 \pi ik)}{(\ln(y)+2 \pi ik)^{j+m}}\right| = \\
&\frac 1 {|\ln(y)+2\pi i k|} \sum_{m=1}^\infty  \left|\frac{\ln(\ln(y)+2 \pi ik)}{\ln(y)+2 \pi ik}\right|^m  \sum_{j=1}^\infty  \frac {|c_{j,m}|} {|\ln(y)+2\pi i k|^{j-1}}.
\end{split}
\end{equation}
Since the double sum is absolutely convergent for large $y$ (see page 349 of \cite{cghjk}), for every large $y$ the terms
\[
 \left|\frac{\ln(\ln(y)+2 \pi ik)}{\ln(y)+2 \pi ik}\right|^m  \sum_{j=1}^\infty  \frac {|c_{j,m}|} {|\ln(y)+2\pi i k|^{j-1}}
\]
tend to zero as $m \to \infty$. Moreover, as $y$ increases, the terms get smaller. Hence, there is $N$ such that, setting $y=N$, for every $m \ge 1$ we have
\[
\left|\frac{\ln(\ln(N)+2 \pi ik)}{\ln(N)+2 \pi ik}\right|^m  \sum_{j=1}^\infty  \frac {|c_{j,m}|} {|\ln(N)+2\pi i k|^{j-1}} \le N.
\]
Plugging that into \eqref{e:jge1} gives
\[
\begin{split}
\sum_{m=1}^\infty\sum_{j=1}^\infty \left| c_{j,m} \frac{\ln^{m}(\ln(y)+2 \pi ik)}{(\ln(y)+2 \pi ik)^{j+m}}\right|  &\le \frac N {|\ln(y)+2\pi i k|} \sum_{m=1}^\infty  \left|\frac{\ln(\ln(y)+2 \pi ik)}{\ln(y)+2 \pi ik}\cdot \frac{\ln(N)+2 \pi ik}{\ln(\ln(N)+2 \pi ik)}\right|^{m} \\
& \le \frac {2N|\ln(N)+2 \pi ik|}{|\ln(\ln(N)+2 \pi ik)|}\cdot \frac{|\ln(\ln(y)+2 \pi ik)|}{|\ln(y)+2\pi i k|^2},
\end{split}
\]
for $y$ large enough. If $y$ is large enough, then 
\[
\frac {2N|\ln(N)+2 \pi ik|}{|\ln(\ln(N)+2 \pi ik)|} \le |\ln(\ln(y)+2 \pi ik)|,
\]
which implies
\[
\sum_{m=1}^\infty\sum_{j=1}^\infty \left| c_{j,m} \frac{\ln^{m}(\ln(y)+2 \pi ik)}{(\ln(y)+2 \pi ik)^{j+m}}\right|  \le \left|\frac{\ln(\ln(y)+2 \pi ik)}{\ln(y)+2 \pi ik}\right|^2.
\]
Plugging this and \eqref{e:j=0} into \eqref{e:twosums} gives \eqref{e:tailest}.
\end{proof}

\section{Proofs of Theorems}

We rewrite \eqref{e:reseq} as
\[
(-2izh^{-1} + h^{-\alpha})e^{-2izh^{-1} +  h^{-\alpha}}  = h^{-\alpha} e^{h^{-\alpha}} ,
\]
and then solve for $z$ using \eqref{e:wk} with
\[
x=-2izh^{-1} + h^{-\alpha},
\]
and 
\[
y=h^{-\alpha} e^{h^{-\alpha}}.
\]
Thus the solutions to \eqref{e:reseq} are given by
\begin{equation}\label{e:zk}
\begin{split}
z_k&= \frac {ih}2(W_k(y) - h^{-\alpha}) \\
 &= \frac {ih}{2}( \ln(y)+2\pi ik-\ln(\ln(y)+2\pi ik)- h^{-\alpha} + R_k) \\
 &= \frac {ih}{2}( 2\pi i k + \ln(h^{-\alpha})-\ln(\ln(y)+2\pi ik)  + R_k).\\
 &= \frac {ih}{2}\left( 2\pi i k -\ln\left(\frac{\ln(y)+2\pi ik}{h^{-\alpha}}\right)  + R_k\right),
\end{split}
\end{equation}
where $k$ varies over $\mathbb Z$. At this point, let us mention that Herbst and Mavi \cite{hm} and Sacchetti \cite{s} use a version of this equation to analyze resonances for the problem $-\partial_x^2 + \alpha \delta_a$ with $\alpha \to \infty$, but with $a>0$ and $k$ fixed. Thus, Herbst and Mavi's and Sacchetti's analysis would correspond to the asymptotic regime $z \to 0$ in our paper.

Below we keep the subscript $k$ to emphasize the role of this integer in our estimates, but note that $z_k$ is the same  as the $z$ as in the statements of the Theorems. In our next lemma we show that the first term in the right hand side of \eqref{e:zk} is the dominant term, and thus that $k$ is roughly of size $h^{-1}$.

\begin{lem}\label{l:keps}
For $k$ such that $z_k$ given by \eqref{e:zk} obeys \eqref{e:zeps}, we have 
\begin{equation}\label{e:keps}
\varepsilon / 2 \le \pi |k| h \le  2/\varepsilon,
\end{equation}
for $h$ small enough.
\end{lem}

\begin{proof}
We begin with the proof of the first inequality of \eqref{e:keps}. Assume for the sake of contradiction that $\pi |k| h <  \varepsilon / 2$. By \eqref{e:tailest} and \eqref{e:lnlnln} we have $|R_k| \le 1$. Combining the first inequality of \eqref{e:zeps} with \eqref{e:zk}, and plugging in $|R_k| \le 1$,  gives
\[
\varepsilon \le |z_k| \le \pi h |k| +  \frac h 2 \left|\ln\left(\frac{\ln(y)+2\pi ik}{h^{-\alpha}}\right) \right| + \frac h2.
\]
Then, using $\pi |k| h <  \varepsilon / 2$, we have
\[
\varepsilon \le h \left|\ln\left(\frac{\ln(y)+2\pi ik}{h^{-\alpha}}\right) \right| + h.
\]
If $h$ is small enough, this implies
\[
\frac \varepsilon 2 \le h \left|\ln|\ln(y)+2\pi ik| \right|,
\]
or
\[
e^{\varepsilon/(2h)} \le|\ln(y)+2\pi ik| \le h^{-\alpha} + \ln h^{-\alpha} + \varepsilon/h,
\]
which is a contradiction (for $h$ small enough).

For the proof of the second inequality of \eqref{e:keps}, assume for the sake of contradiction that $\pi |k| h >  2/\varepsilon $. Combining the first inequality of \eqref{e:zeps} with \eqref{e:zk}, and plugging in $|R_k| \le 1$,  gives
\[
\pi h |k| -  \frac h 2 \left|\ln\left(\frac{\ln(y)+2\pi ik}{h^{-\alpha}}\right) \right| - \frac h2 \le |z_k| \le 1/ \varepsilon.
\]
If $h$ is small enough, this implies
\[
\pi h |k| - \frac h 2\left|\ln|\ln(y)+2\pi ik| \right| \le \frac 3{2 \varepsilon},
\]
or
\[
 \exp\left( 2 \pi |k| - \frac 3{h \varepsilon}\right) - 2 \pi |k| \le \ln (y),
\]
Using now  $\pi |k| h >  2/\varepsilon $ gives
\begin{equation}\label{e:expln}
 \exp\left( \frac{4 } {h \varepsilon } - \frac 3{h \varepsilon}\right) - \frac 4 {h \varepsilon} \le \ln (y) = h^{-\alpha} + \ln(h^{-\alpha});
\end{equation}
indeed, since the function $f(x):=e^{x-a} - x$ is increasing on $(a,\infty)$, we have $f(2\pi|k|)>f(4/h\varepsilon)$ with $a=3/h\varepsilon$. Finally, \eqref{e:expln} is a contradiction for $h$ small enough.
\end{proof}

We now proceed to the proofs of the theorems. We simplify the logarithm on the right hand side of \eqref{e:zk} differently for $\alpha<1$ and $\alpha>1$, according to which term on the inside is larger.

\begin{proof}[Proof of Theorem \ref{t:asmall}]
Since $\alpha \in (0,1)$, the first inequality of \eqref{e:keps} implies that 
\begin{equation}\label{e:khlny}
|2\pi i k|  \ge \varepsilon h^{-1} >\ln(h^{-\alpha}) + h^{-\alpha}  =  \ln(y),
\end{equation}
for $h$ small enough, and we write
\[
\ln\left(\frac{\ln(y)+2\pi ik}{h^{-\alpha}}\right) = \ln( 2 \pi i k h^\alpha) + \ln\left(1 + \frac{\ln(y)}{2\pi i k}\right),
\]
which, inserted into \eqref{e:zk}, gives
\[
z_k = \frac{ih}{2}(2\pi ik-\ln( 2 \pi i k h^\alpha) + R'_k),
\qquad \text{where} \qquad 
R_k' = R_k - \ln\left(1 + \frac{\ln(y)}{2\pi i k}\right).
\]
Using $|R_k| \le 1$ as in Lemma 2, and \eqref{e:khlny}, we get
\begin{equation}\label{e:rk'}
|R_k'| \le |R_k| + \ln\left|1 + \frac{\ln(y)}{2\pi i k}\right| + \Big|\arg\left(1 + \frac{\ln(y)}{2\pi i k}\right)\Big|\le 1 + \ln 2  + \pi/2 \le 4.
\end{equation}
Hence
\begin{equation}\label{e:rezkimzk}\begin{split}
\re z_k &= - \pi h k + \frac h 2 \arg(2\pi i k h^\alpha) - \frac h2 \im R_k',\\
\im z_k & = - \frac h 2 \ln|2\pi k h ^\alpha| + \frac h 2 \re R_k'.
\end{split}\end{equation}
From these equations, we can already see the leading order behavior of the resonances. 

To get a precise statement with good remainder estimates, we manipulate the equations \eqref{e:rezkimzk} in the following way. Taking real and imaginary parts of \eqref{e:reseq} gives
\[\begin{split}
h^{-\alpha} e^{-2\im z_k/h} \cos(2\re z_k/h) &= h^{-\alpha} + 2\im z_k h^{-1}  \\
h^{-\alpha} e^{-2\im z_k/h} \sin(2\re z_k/h) &= - 2\re z_k h^{-1}.
\end{split}\]
Squaring the equations and adding them using $\cos^2 + \sin^2 = 1$ gives
\[
\frac{e^{4\im z_k/h}(4(\re z_k)^2 + 4(\im z_k)^2)}{h ^ {2-2\alpha}} + \frac{e^{4\im z_k/h} (4 \im z_k)}{h^{1-\alpha}} + e^{4\im z_k/h} = 1
\]
or
\[\begin{split}
-\frac {4\im z_k}h &= \ln\left(1 + 4\im z_k h^{\alpha -1} + (4(\re z_k) ^2 + 4(\im z_k)^2) h^{2\alpha -2}\right)\\
&= \ln(4 (\re z_k) ^2 h^{2\alpha -2}) + \ln(1+t),
\end{split}\]
where
\[
t=\frac{4(\im z_k)^2 + 4 \im z_k h^{1-\alpha} +   h^{2-2\alpha}}{ 4(\re z_k) ^{2} }.
\]
Combine with $0 \le \ln(1+t) \le t$ 
to get
\begin{equation}\label{e:ln0t}
0 \le -\frac {4\im z_k}h - \ln(4(\re z_k)^2 h^{2\alpha-2}) \le t.
\end{equation}
From \eqref{e:rezkimzk} we have, using \eqref{e:keps}, \eqref{e:rk'}, and $|\arg(2\pi i k h^\alpha)| = \pi/2$, that
\[
|\re z_k| \ge \frac \varepsilon 2 - \frac {h\pi}4 - 2 h \ge \varepsilon /3, \qquad \text{and} \qquad 
|\im z_k| \le  h\ln(4\varepsilon^{-1}h^{\alpha-1}).
\]
That gives
\[
t \le \frac{4h^2\ln^2(4\varepsilon^{-1}h^{\alpha-1}) + 4 h\ln(4\varepsilon^{-1}h^{\alpha-1}) h^{1-\alpha} +   h^{2-2\alpha}}{ 4(\re z_k) ^{2} }\le \frac{h^{2-2\alpha}}{2 (\re z_k)^2} \le 5 h^{2-2\alpha}\varepsilon^{-2}.
\]
Plugging into \eqref{e:ln0t} gives \eqref{e:tasmall}.
\end{proof}

\begin{proof}[Proof of Theorem \ref{t:abig}]
Since $\alpha>1$,  the second inequality of \eqref{e:keps} implies that
\[
|2\pi i k| \le 4 \varepsilon^{-1} h^{-1} < \ln(h^{-\alpha}) + h^{-\alpha}  = \ln(y),
\]
for $h$ small enough, and we write
\[\begin{split}
\ln\left(\frac{\ln(y)+2\pi ik}{h^{-\alpha}}\right) &= \ln \left( 1 + h^\alpha \ln(h^{-\alpha}) + 2 \pi i k h^\alpha\right),
\end{split}\]
which, inserted into \eqref{e:zk}, gives
\begin{equation}\label{e:rzkizk}\begin{split}
\re z_k &=  \frac {h}{2}\left( -2\pi  k + \im \ln \left( 1 + h^\alpha \ln(h^{-\alpha}) + 2 \pi i k h^\alpha\right)  - \im R_k\right),\\
\im z_k &= \frac {h}{2}\left( - \ln \left| 1 + h^\alpha \ln(h^{-\alpha}) + 2 \pi i k h^\alpha\right|  + \re R_k\right). 
\end{split}\end{equation}
This time the manipulations between \eqref{e:rezkimzk} and \eqref{e:ln0t} are not necessary, but we must expand further to get a nonvanishing imaginary part. We will use the approximation $\ln(1+t) \sim t$. To bound the remainder, we write
\[
\ln \left| 1 + h^\alpha \ln(h^{-\alpha}) + 2 \pi i k h^\alpha\right| =  \tfrac 12 \ln(1 + 2h^\alpha \ln(h^{-\alpha}) + h^{2\alpha}\ln^2(h^{-\alpha}) + 4 \pi^2 k^2 h^{2\alpha}),
\]
and use the estimate $0 \le  - \ln(1+t) + t  \le t^2/2$, with 
\begin{equation}\label{e:tdef}
t = 2h^\alpha \ln(h^{-\alpha}) + h^{2\alpha}\ln^2(h^{-\alpha}) + 4 \pi^2 k^2 h^{2\alpha},
\end{equation}
to write, using $\pi^2k^2 \le 4 \varepsilon^{-2}h^{-2}$,
\[\begin{split}
- h^{2\alpha}\ln^2(h^{-\alpha}) &\le -\ln(1 +  2h^\alpha \ln(h^{-\alpha}) + h^{2\alpha}\ln^2(h^{-\alpha}) + 4 \pi^2 k^2 h^{2\alpha}) +  2h^\alpha \ln(h^{-\alpha}) +  4 \pi^2 k^2 h^{2\alpha} \\
&\le \frac {t^2} 2 - h^{2\alpha} \ln^2(h^{-\alpha})\\
&= \frac 12(2h^\alpha \ln(h^{-\alpha}) + h^{2\alpha}\ln^2(h^{-\alpha}) + 4 \pi^2 k^2 h^{2\alpha})^2 - h^{2\alpha} \ln^2(h^{-\alpha})\\
&\le 130 \varepsilon^{-4}h^{4\alpha - 4}+ 2 h^{2\alpha}\ln^2(h^{-\alpha}),
\end{split}\]
which implies
\begin{equation}\label{e:lnt}
\Big| -\ln \left| 1 + h^\alpha \ln(h^{-\alpha}) + 2 \pi i k h^\alpha\right|  +  h^\alpha \ln(h^{-\alpha}) +  2 \pi^2 k^2 h^{2\alpha} \Big| \le 65 \varepsilon^{-4}h^{4\alpha - 4}+  h^{2\alpha}\ln^2(h^{-\alpha}).
\end{equation}
where the first term on the right is dominant when $\alpha < 2$ and the second term is dominant when $\alpha \ge2$.

Inserted into \eqref{e:rzkizk}, that gives
\begin{equation}\label{e:imk1}\begin{split}
|\im z_k +   \pi^2 k^2 h^{2\alpha+1}| &= \frac h2 \left| - \ln \left| 1 + h^\alpha \ln(h^{-\alpha}) + 2 \pi i k h^\alpha\right| + \re R_k  + 2 \pi^2 k^2 h^{2\alpha} \right| \\ & \le \frac h 2 \left(|\re R_k - h^\alpha \ln(h^{-\alpha})| + 65 \varepsilon^{-4}h^{4\alpha - 4}+  h^{2\alpha}\ln^2(h^{-\alpha})\right),
\end{split}\end{equation}
where for the inequality we added and subtracted $h^\alpha \ln(h^{-\alpha})$, and used \eqref{e:lnt}.

To deal with $R_k$ term, using $\ln(y) = h^{-\alpha} + \ln(h^{-\alpha})$ and $|2\pi i k| \le 4 \varepsilon^{-1} h^{-1}$, we have
\begin{equation}\label{e:lbd}\begin{split}
\left|\frac{\ln(\ln(y) + 2 \pi i k)}{\ln(y)+2\pi i k}\right|&= \sqrt{\frac{\ln^2|\ln(h^{-\alpha})+h^{-\alpha}+2\pi ik|+ (\arg(\ln(h^{-\alpha})+h^{-\alpha}+2\pi ik))^2}{(\ln(h^{-\alpha})+h^{-\alpha})^2+(2\pi k)^2}}\\
&\le \frac {\ln(h^{-\alpha})}{h^{-\alpha}} \frac {|\ln|\ln(y) + 2 \pi i k| + i \pi/2|/\ln(h^{-\alpha})} {|1 + h^\alpha \ln(h^{-\alpha}) + h^\alpha 2 \pi i k|} \le 2 h^\alpha \ln (h^{-\alpha}),
\end{split}\end{equation}
for $h$ sufficiently small, 
so, by \eqref{e:tailest},
\begin{equation}\label{e:rk1}
\left|R_k - \frac{\ln(\ln(y) + 2 \pi i k)}{\ln(y)+2\pi i k}\right| \le 8 h^{2\alpha} \ln^2 (h^{-\alpha}).
\end{equation}
Next,
\[
\re \frac{\ln(\ln(y) + 2 \pi i k)}{\ln(y)+2\pi i k} - h^\alpha \ln h^{-\alpha} =  (h^\alpha \ln h^{-\alpha}) \left(\frac {s-t} {1+t} \right),
\]
where $t$ is as in \eqref{e:tdef} and
\[\begin{split}
s &= h^\alpha \ln h^{-\alpha} + \frac 12 \frac{\ln(1+t)}{\ln(h^{-\alpha})} + \frac 12 h^\alpha \ln(1+t) +   \frac{2\pi k h^\alpha\arg(\ln y + 2 \pi i k)}{\ln|\ln y+2\pi i k|} \left(1+\frac{\ln(1+t)}{2\ln(h^{-\alpha})}\right) \\
& \le \frac t 2 + \frac t {2 \ln (h^{-\alpha})} + \frac 1 2 h^\alpha t + \frac {2t}{\ln(h^{-\alpha})} \le t,
\end{split}
\]
for $h$ small enough.
That implies
\[
-t  h^\alpha \ln h^{-\alpha}\le \re \frac{\ln(\ln(y) + 2 \pi i k)}{\ln(y)+2\pi i k} - h^\alpha \ln h^{-\alpha}  \le  0,
\]
which implies
\begin{equation}\label{e:rk2}
\left|\re \frac{\ln(\ln(y) + 2 \pi i k)}{\ln(y)+2\pi i k} - h^\alpha \ln h^{-\alpha} \right| \le 2 h^{2\alpha}\ln^2(h^{-\alpha}) + h^{3\alpha}\ln^3(h^{-\alpha}) + 4 \pi^2 k^2 h^{3\alpha}\ln h^{-\alpha}.
\end{equation}
Substituting \eqref{e:rk1} and \eqref{e:rk2} into \eqref{e:imk1} gives
\[\begin{split}
|\im z_k +   \pi^2 k^2 h^{2\alpha+1}| & \le \frac h 2(11 h^{2\alpha} \ln^2(h^{-\alpha}) + h^{3\alpha}\ln^3 (h^{-\alpha}) + 16 \varepsilon^{-2} h^{3\alpha-2} \ln (h^{-\alpha}) + 65 \varepsilon^{-4} h^{4\alpha-4})\\
&\le 6 h^{2\alpha+1} \ln^2(h^{-\alpha})+ 33 \varepsilon^{-4}h^{4\alpha - 3}.
\end{split}\]
Returning to the real part, from \eqref{e:rzkizk}, and using again \eqref{e:lbd} and \eqref{e:rk1}, we have
\[\begin{split}
|\re z_k + \pi k h| &=  \frac h 2|  \arg \left( 1 + h^\alpha \ln(h^{-\alpha}) +  2\pi i k h^{\alpha}\right)  -  \im R_k|\\
&\le \pi k h^{\alpha+1} + 4 h^{2\alpha+1}\ln^2h^{-\alpha} + h^{\alpha+1} \ln h^{-\alpha}  \le 3 \varepsilon^{-1} h^\alpha,
\end{split}\]
which implies, after factoring and using also \eqref{e:zeps},
\[
|\pi^2k^2h^{2\alpha+1} - (\re z_k)^2 h^{2\alpha -1}| \le 3 \varepsilon^{-1} h^{2\alpha} |\pi k h^{\alpha} - \re z_k h^{\alpha-1}|  \le 9 \varepsilon^{-2} h^{3\alpha-1}, 
\]
and hence
\[
|\im z_k + (\re z_k)^2 h^{2\alpha -1}|\le  7 h^{2\alpha+1} \ln^2(h^{-\alpha}) +  34 \varepsilon^{-4}h^{4\alpha - 3}.
\]
\end{proof}

\noindent\textbf{Concluding remarks.} At the beginning of the proof of Theorem \ref{t:asmall}, we used \eqref{e:khlny} to say that $|2\pi i k|$ always dominates $\ln(y)$ when $\alpha \in (0,1)$ whenever $\varepsilon < |z| < 1/\varepsilon$. At the beginning of the proof of Theorem \ref{t:abig}, we similarly saw that $\ln(y)$ always dominates $|2\pi i k|$ when $\alpha >1$ for this range of $z$. But when $\alpha = 1$, neither of these terms dominates the other for all $z$ in this range. It would be necessary to distinguish large and small $z$ using a finer analysis in this case.

\end{document}